\documentclass{llncs}
\usepackage{amsmath,amssymb,amscd,verbatim,enumerate,graphicx,subfigure,array}
\usepackage{wrapfig}
\usepackage[utf8]{inputenc}
\usepackage[colorlinks=true,citecolor=black,
linkcolor=black,urlcolor=blue]{hyperref}

\author{Jean Cardinal\inst{1} \and Stefan Felsner\inst{2}}
\institute{
Universit\'e libre de Bruxelles (ULB), 
Brussels, Belgium.\\
\email{jcardin@ulb.ac.be}
\and
Technische Universit\"at Berlin, 
Berlin, Germany.\\
\email{felsner@math.tu-berlin.de}
}
\title{Covering Partial Cubes with Zones}
\date{}
\begin{document}
\maketitle
\sloppy

\begin{abstract}
  A partial cube is a graph having an isometric embedding in a
  hypercube. Partial cubes are characterized by a natural equivalence
  relation on the edges, whose classes are called {\em zones}. The
  number of zones determines the minimal dimension of a hypercube
  in which the graph can be embedded. We consider the problem of
  covering the vertices of a partial cube with the minimum number of
  zones. The problem admits several special cases, among which are the
  problem of covering the cells of a line arrangement with a minimum
  number of lines, and the problem of finding a minimum-size fibre in
  a bipartite poset. For several such special cases, we give upper and
  lower bounds on the minimum size of a covering by zones. We also
  consider the computational complexity of those problems, and
  establish some hardness results.
\end{abstract}

\section{Introduction}

As an introduction and motivation to the problems we consider, let us
look at two puzzles.
\medskip

\begin{wrapfigure}[16]{r}{0.44\textwidth}%
\vskip-9mm
\kern5pt\includegraphics[width=0.45\textwidth]{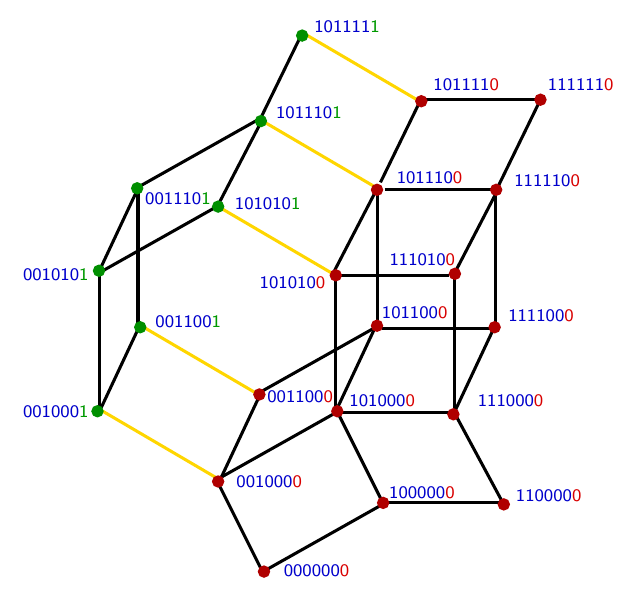}
\caption{\label{fig:pcube} A partial cube with vertex
labels representing an isometric embedding in $Q_7$ 
and a highlighted zone.}
\end{wrapfigure}
\noindent
{\bf Hitting a consecutive pair.} Given a set of $n$ elements, how
many pairs of them
 must be chosen so that every permutation of the $n$
elements has two consecutive elements forming a chosen pair?
\medskip

\noindent
{\bf Guarding cells of a line arrangement.} Given an arrangement of
$n$ straight lines in the plane, how many lines must be chosen so that
every cell of the arrangement is bounded by at least one of the chosen
line?
\vskip8pt

\noindent
While different, the two problems can be cast as special cases of a
general problem involving {\em partial cubes}.  The $n$-dimensional
hypercube $Q_n$ is the graph with the set $\{0,1\}^n$ of binary words
of length $n$ as vertex set, and an edge between every pair of
vertices that differ on exactly one bit. A~{\em partial cube} $G$ of
dimension $n$ is a subgraph of the $n$-dimensional hypercube with the
property that the distance between two vertices in $G$ is equal to
their Hamming distance, i.e., their distance in $Q_n$. In general, a
graph $G$ is said to have an {\em isometric embedding} in another
graph $H$ whenever $G$ is a subgraph of $H$ and the distance between
any two vertices in $G$ is equal to their distance in $H$. Hence
partial cubes are the graphs admitting an isometric
embedding in $Q_n$, for some $n$. The edges of
a partial cube can be partitioned into at most $n$ equivalence classes
called {\em zones}, each corresponding to one of the $n$ directions of
the edges of $Q_n$.

We consider the following problem:
\medskip

\noindent
{\bf Partial Cube Covering:} {\it Given a partial cube, find a
  smallest subset $S$ of its zones such that every vertex is incident
  to an edge of one of the zones in $S$.}
\medskip

If we refer to the labeling of the vertices by words in $\{0,1\}^n$,
the problem amounts to finding a smallest subset $I$ of $[n]$ such
that for every vertex $v$ of the input partial cube, there is at least
one $i\in I$ such that flipping bit $i$ of the word of $v$ yields
another vertex.

\section{Classes of Partial Cubes}

The reader is referred to the books of Ovchinnikov~\cite{O11} and
Imrich and Klav\v{z}ar~\cite{IK00} for known results on partial cubes.
Let us also mention that some other structures previously defined in
the literature are essentially equivalent to partial cubes. Among them
are {\it well-graded families of sets} defined by Doignon and
Falmagne~\cite{DF11}, and {\it Media}, defined by Eppstein, Falmagne, and
Ovchinnikov~\cite{EFO08}.

We are interested in giving bounds on the minimum number of zones
required to cover the vertices of a partial cube, but also in the
computational complexity of the problem of finding such a minimum
cover.  Regarding bounds, it would have been nice to have a general
nontrivial result holding for every partial cube. Unfortunately, only
trivial bounds hold in general.

In one extreme case, the partial cube $G=(V,E)$ is a star, consisting
of one central vertex connected to $|V|-1$ other vertices of degree
one. This is indeed a partial cube in dimension $n=|V|-1$, every zone
of which consists of a single edge. Since there are $n$ vertices of
degree one, all $n$ zones must be chosen to cover all vertices.  In
the other extreme case, the partial cube is such that there exists a
single zone covering all vertices. This lower bound is attained by the
hypercube $Q_n$.

Table~\ref{tab:res} gives a summary of our results for the various families of
partial cubes that we considered. For each family, we consider upper and lower
bounds and complexity results. We now briefly describe the various families we
studied. The proofs of the new results are in the following sections.

\subsection{Hyperplane Arrangements} 

The dual graph of a simple arrangement of $n$ hyperplanes in
$\mathbb{R}^d$ is a partial cube of dimension $n$. The partial cube
covering problem becomes the following: given a simple hyperplane
arrangement, find a smallest subset $S$ of the hyperplanes such that
every cell of the arrangement is bounded by at least one hyperplane in
$S$. 

\begin{wrapfigure}[10]{r}{0.44\textwidth}%
\vskip-10mm
\includegraphics[width=0.44\textwidth]{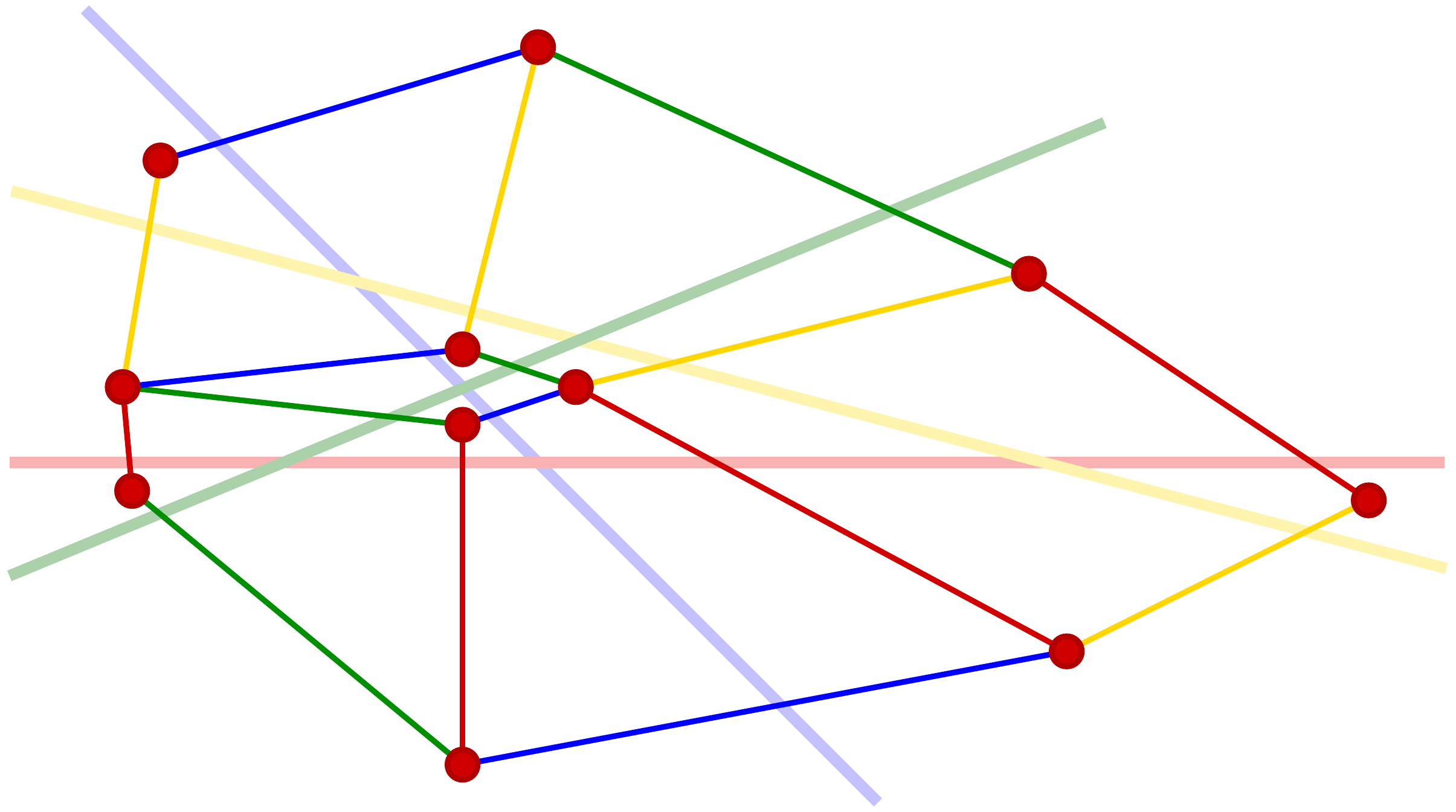}
\caption{An arrangement of 4 lines superimposed with its dual.} 
\label{fig:larr}
\end{wrapfigure}

The case of line arrangements was first considered
in~\cite{BCHKLT12}. In fact it was this problem that motivated us to
investigate the generalization to partial cubes. The best known lower and
upper bounds for line arrangements on $|S|$ are of order $n/6$
and $n-O(\sqrt{n\log n})$ (see~\cite{AP12}). The
complexity status of the optimization problem is unknown.

Instead of lines and hyperplanes it is also possible to consider 
Euclidean or spherical arrangements of pseudo-lines and pseudo-hyperplanes,
their duals are still partial cubes. Actually, all the cited results apply
to the case of arrangements of pseudo-lines.

Spherical arrangements of pseudo-hyperplanes are equivalent to
oriented matroids, this is the Topological Representation Theorem of
Folkman and Lawrence. The pseudo-hyperplanes correspond to the elements of the
oriented matroid and the cells of the arrangement correspond to the
topes of the oriented matroid. Hence our covering problem
asks for a minimum size set $C$ of elements such that for every tope $T$
there is an element $c\in C$ such that $T \oplus c$ is another tope.
For more on oriented matroids we refer to~\cite{blwsz-om-93}.    

\subsection{Acyclic Orientations} 

\begin{wrapfigure}[15]{r}{0.44\textwidth}%
\vskip-16mm
\includegraphics[width=0.44\textwidth]{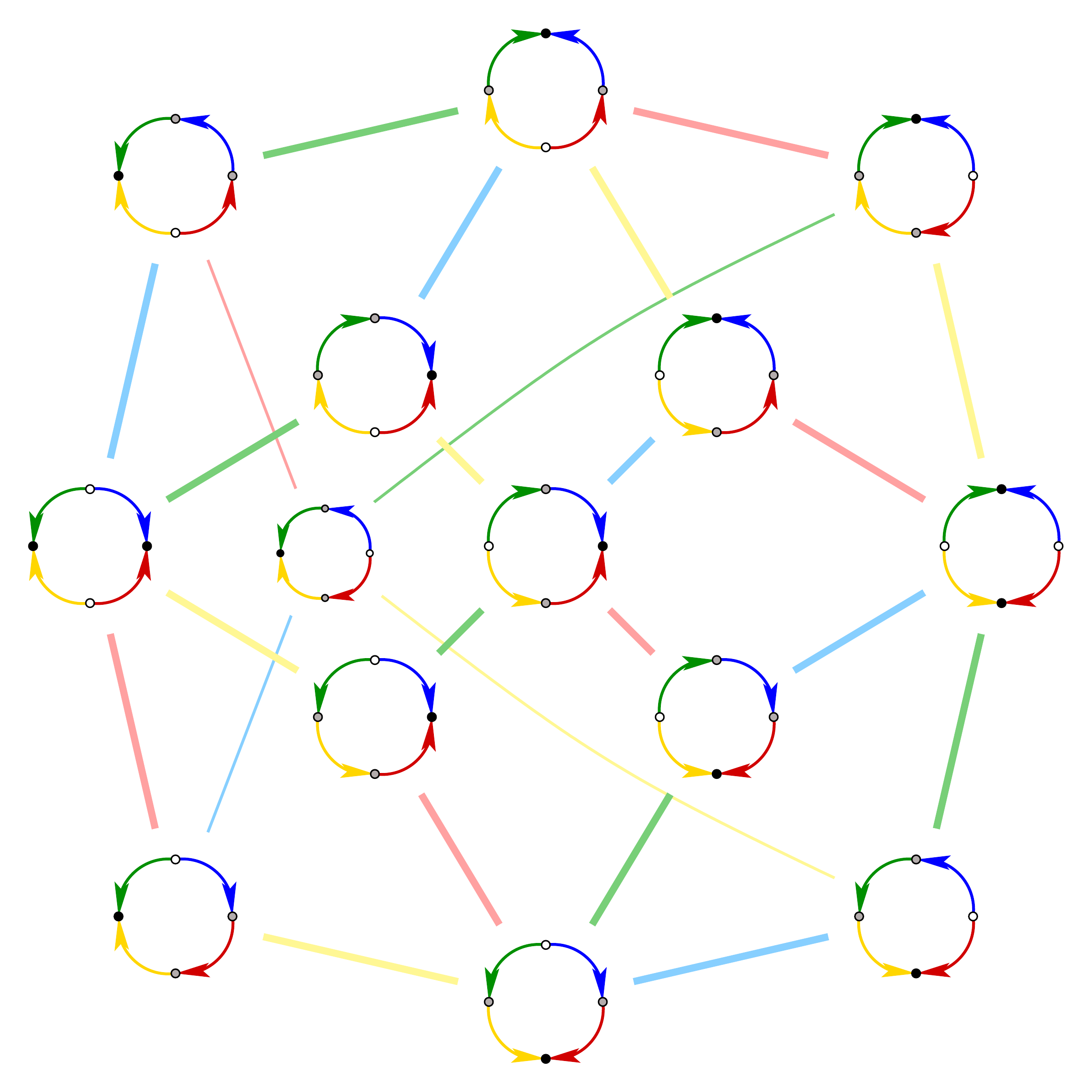}
\caption{Flip graph of acyclic orientations of the 4-cycle.} 
\label{fig:aoC4}
\end{wrapfigure}
From a graph $G=(V,E)$, we can define a partial cube $H$ in which
every vertex is an acyclic orientation of $G$, and two orientations
are adjacent whenever they differ by a single edge reversal ({\it flip}). 
This partial cube has dimension equal to the number of edges
of $G$. It is also the dual graph of the arrangement of the $|E|$
hyperplanes of equation $x_i = x_j$ for $ij\in E$ in
$\mathbb{R}^{|V|}$. Every cell of this arrangement corresponds to an
acyclic orientation, and adjacent cells are exactly those that differ
by a single edge. This observation, in particular, shows that $H$ is
connected.
The graph $H$ and the notion of edge flippability have been studied by
Fukuda et al.~\cite{FPS01} and more recently by Cordovil and
Forge~\cite{CF12}.

The partial cube covering problem now becomes the following: given a
graph $G=(V,E)$, find a smallest subset $S\subseteq E$ such that for
every acyclic orientation of $G$, there exists $e\in S$ such that
flipping the orientation of $e$ does not create a cycle.

\subsection{Median Graphs}

A median graph is an undirected graph in which every three vertices
$x$, $y$, and $z$ have a unique {\it median}, i.e., a vertex
$\mu(x,y,z)$ that belongs to shortest paths between each pair of $x$,
$y$, and $z$. Median graphs form structured subclass of partial cubes
that has been studied extensively, see~\cite{Chen12} and references
therein. Since the star and the hypercube are median graph there are
no non-trivial upper or lower bounds for the zone cover problem on
this class. We use a construction of median graphs 
to prove hardness of the zone cover problem.

\subsection{Distributive Lattices} 

\begin{wrapfigure}[12]{r}{0.44\textwidth}%
\vskip-9mm
\includegraphics[width=0.44\textwidth]{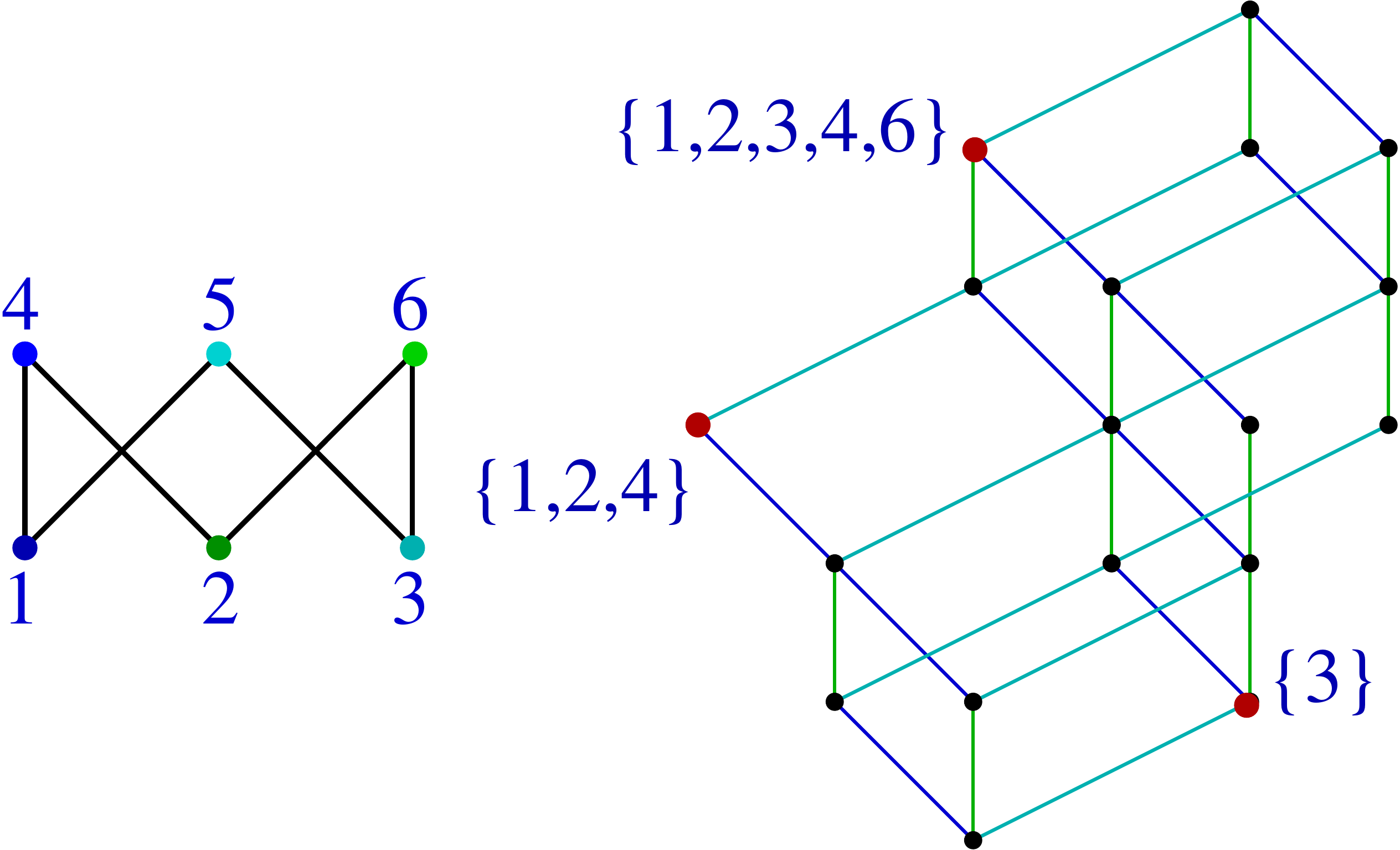}
\caption{A poset (left) and its lattice of down-sets (right).} 
\label{fig:poset}
\end{wrapfigure}
Cover graphs of distributive lattices are partial cubes. From
Birkhoff's representation theorem (a.k.a.~Fundamental Theorem of
Finite Distributive Lattices) we know that there is a poset $P$ such that 
the vertices of the partial cube are the downsets of $P$, the zones 
 of the partial cube in turn correspond to the elements of $P$.

The problem becomes: given a poset $P$, find a smallest subset
$S$ of its elements such that for every downset $D$ of $P$, there
exists $v\in S$ such that either $D\cup\{v\}$ or $D\setminus \{v\}$ is
a downset, distinct from $D$.

\subsection{Trees} 

Trees with $n$ edges are partial cubes of dimension exactly
$n$. Since every zone contains exactly one edge, the
partial cube covering problem on trees boils down to the edge cover
problem on trees. There are instances attaining the lower and upper bounds of
$n/2$ and $n-1$, moreover, there is a simple dynamic programming algorithm 
that computes an optimal cover in linear time. 

\begin{table}
\begin{center}
\footnotesize
\newcolumntype{L}[1]{>{\small\raggedright\arraybackslash}p{#1}}
\newcolumntype{C}[1]{>{\small\centering\arraybackslash}p{#1}}
\begin{tabular}{|L{.30\textwidth}|C{.14\textwidth}|C{.15\textwidth}|L{.31\textwidth}|}
\hline
Partial Cube & Lower bound & Upper bound & Complexity \\
\hline
\hline
Arrangements of $n$ lines\break
in $\mathbb{R}^2$ &
\vskip-5pt $n-o(n)$
\hbox{(Thm. \ref{thm:exlb})} &
\vskip-5pt
\hbox{$n-\Omega (\sqrt{n\log n})$} \cite{AP12} &
\vskip-3pt {\hfil -- \hfil} \\
\hline
Arrangements of $n$ hyperplanes in $\mathbb{R}^d$ &
\vskip-3pt -- &
\vskip-3pt $n-\Omega (n^{1/d})$ 
\hbox{(Thm. \ref{thm:ubhyper})}&
\vskip-3pt {\hfil -- \hfil} \\
\hline
\vskip-5pt Acyclic orientations &
\vskip2pt -- & Minimum edge cut (Thm.~\ref{thm:cut}) &
Recognition is coNP-complete, even for complete graphs (Thm.~\ref{thm:recgs}) \\
\hline
\vskip-5pt Median graphs &
\vskip-3pt 1 &
\vskip-3pt $n$ & NP-complete (Cor.~\ref{cor:medfromtri}), APX-hard (Cor.~\ref{cor:medfromtriapx}) \\
\hline
Distributive lattices with representative poset of $n$ elements &
\vskip2pt -- &
\vskip-3pt $2n/3$ (Cor.~\ref{cor:fibre}) &
Recognition is coNP-complete (Cor.~\ref{cor:recposet})
 $\Sigma_2^P$-complete (Cor.~\ref{cor:findposet}) \\
\hline
Trees with $n$ edges & $n/2$ & $n-1$ & P (min edge cover)\\
\hline
\end{tabular}
\end{center}
\caption{\label{tab:res}Worst-case bounds and complexities for the special cases of the partial cube covering problem. When used, $n$ denotes the dimension of the partial cube.}
\end{table}
\section{Line Arrangements}

We first give a lower bound on the size of a guarding set for an arrangement of lines.

\begin{theorem}
\label{thm:exlb}
The minimum number of lines needed to guard the cells of any arrangement of $n$ lines is $n-o(n)$.
\end{theorem}
\begin{proof}
The proof is a direct consequence of known results on the following problem from Erd\H{o}s:
given a set of points in the plane with no four points on a line, find the largest subset in general position, that is, with no three points on a line. 
Let $\alpha (n)$ be the minimum size of such a set over all arrangements of $n$ points. 
F\"uredi observed~\cite{F91} that $\alpha (n) = o(n)$ follows from the density version of the Hales-Jewett theorem~\cite{FK89}. 
But this directly proves that we need at least $n - o(n)$ lines to guard all cells of an arrangement. 
The reduction is the one observed by Ackerman et al.~\cite{AP12}: consider the line arrangement that is dual to the point set, and slightly perturb it so 
that each triple of concurrent lines forms a cell of size three. Now the complement of any guarding set is in general position in the primal point set.\qed
\end{proof}

We now show that for arrangements of hyperplanes in $\mathbb{R}^d$, with $d=O(1)$, 
there always exists a guarding set of size at most $n-\Omega (n^{1/d})$. This generalizes
a result for $d=2$ obtained in \cite{BCHKLT12}.

\begin{theorem}
\label{thm:ubhyper}
In every arrangement of hyperplanes in general position in $\mathbb{R}^d$ (with $d=O(1)$),
there exists a subset of the hyperplanes of size at most $n-\Omega (n^{1/d})$ such that
every cell is bounded by at least one hyperplane in the subset. 
\end{theorem}
\begin{proof}
We will prove that every such arrangement has an independent set of size $\Omega (n^{1/d})$, where
an independent set is defined as the complement of a guarding set, that is, 
a subset of the hyperplanes such that no cell is bounded by hyperplanes of the subset only.

Let $H$ be a set of $n$ hyperplanes, and consider an arbitrary, inclusionwise maximal independent set $I$.
For each hyperplane $h\in H\setminus I$, there must be a cell $c_h$ of the arrangement that is bounded by a set of hyperplanes $C\cup\{h\}$ with
$C\subseteq I$, since otherwise we could add $h$ to $I$, and $I$ would not be maximal. If $c_h$ has size at least $d+1$, there must be a vertex of
the arrangement induced by $I$ that is also a vertex of $c_h$. We charge $h$ to this vertex. Note that each vertex can only be charged $2^d=O(1)$ times. 
If $c_h$ has size $d$, we charge $h$ to the remaining $(d-1)$-tuple of hyperplanes of $I$ bounding $c_h$. This tuple can also be charged at most $O(1)$
times. Therefore
\begin{eqnarray}
|H\setminus I| = n-|I| & \leq & 2^d \cdot O(|I|^d) + O(|I|^{d-1}) \\
|I| & = & \Omega (n^{1/d}) .
\end{eqnarray}       
\qed\end{proof}

\section{Acyclic Orientations}

Given a graph $G=(V,E)$ we wish to find a subset $S\subseteq E$ such
that for every acyclic orientation of $G$, there exists a {\em
  flippable} edge $e\in S$, that is, an edge $e\in S$ such that
changing the orientation of $e$ does not create any cycle. Let us call
such a set a {\em guarding set} for $G$. Note that an oriented edge
$e=uv$ in an acyclic orientation is flippable if and only if it is not
{\em transitive}, that is, if and only if $uv$ is the only oriented
path from $u$ to $v$.

\subsection{Complexity}

\begin{theorem}
\label{thm:recgs}
Given a graph $G=(V,E)$ and a subset $S\subseteq E$ of its edges, the
problem of deciding whether $S$ is a guarding set is coNP-complete,
even if $G$ is a complete graph.
\end{theorem}
\begin{proof}
  The set $S$ is not a guarding set if and only if there exists an
  acyclic orientation of $G$ in which all edges $e\in S$ are
  transitive.

  Consider a simple graph $H$ on the vertex set $V$, and define $G$ as
  the complete graph on $V$, and $S$ as the set of non-edges of $H$.
  We claim that $S$ is a guarding set for $G$ if and only if $H$ does
  not have a Hamilton path. Since deciding the existence of a
  Hamilton path is NP-complete, this proves the result.

  To prove the claim, first suppose that $H$ has a Hamilton path,
  and consider the acyclic orientation of $G$ that corresponds to the
  order of the vertices in the path.  Then by definition, no edge of
  $S$ is in the path, hence all of them are transitive, and $S$ is not
  a guarding set.  Conversely, suppose that $S$ is not guarding. Then
  there exists an acyclic orientation of $G$ in which all edges of $S$
  are transitive.  The corresponding ordering of the vertices in $V$
  yields a Hamilton path in $H$.\qed
\end{proof}

\subsection{Special cases and an upper bound}

\begin{lemma}
\label{lem:complete}
The minimum size of a guarding set in a complete graph on $n$ vertices is $n-1$.
\end{lemma}
\begin{proof}
  First note that there always exists a guarding set of size $n-1$
  that consists of all edges incident to one vertex.

  Now we need to show that any other edge subset of smaller size is
  not a guarding set. This amounts to stating that every graph with at
  least ${n\choose 2} - (n-2)$ edges has a Hamilton path. To see this,
  proceed by induction on $n$.  Suppose this holds for graphs with
  less than $n$ vertices. Consider a set $S$ of at most $n-2$ edges of
  the complete graph. Let $u,v$ be two vertices with $uv\in
  S$.  One of the two vertices, say $v$, is incident to at most
  $\lfloor (n-2)/2\rfloor$ edges of $S$. Consider a Hamilton path on
  the $n-1$ vertices $\neq v$, which exists by induction
  hypothesis. Then vertex $v$ must have one or two incident edges that do not
  belong to $S$ and connect $v$ to the first, or the last, or two
  consecutive vertices of this path. Hence we can integrate $v$ in
  the Hamilton path.\qed
\end{proof}

Since acyclic orientations of the complete graph $K_n$ correspond to
the permutations of $S_n$, this is in fact the solution to the
puzzle ``hitting a consecutive pair" in the introduction.  The dual
graph of the arrangements of hyperplanes corresponding to the complete
graph is known as the permutohedron.  Hence the above result can also
be stated in the following form.

\begin{corollary}
  The minimum size of a set of zones covering the vertices of the
  $n$-dimensional permutohedron is $n-1$.
\end{corollary}

We now give a simple, polynomial-time computable upper bound on the size of a guarding set.
A set $C\subseteq E$ is an {\em edge cut} whenever the graph $(V,E\setminus C)$ is not connected.

\begin{theorem}
\label{thm:cut}
Every edge cut of $G$ is a guarding set of $G$.
\end{theorem}
\begin{proof}
  Consider an edge cut $C\subseteq E$ and an acyclic orientation
  $A_G$ of $G$. This acyclic orientation can be used to define a
  partial order on $V$. Let us consider a total ordering of $V$ that
  extends this partial order, and pick an edge $e=uv\in C$ that
  minimizes the rank difference between $u$ and $v$. We claim that $e$
  is flippable. Suppose for the sake of contradiction that $e$ is not
  flippable. Then $e$ must be transitive and there exists a directed
  path $P$ in $A_G$ between $u$ and $v$ that does not use
  $e$. Since $C$ is a cut, $u$ and $v$ belong to distinct connected
  components of $(V,E\setminus C)$, and $P$ must use another edge
  $e'\in C$. By definition, the endpoints of $e'$ have a rank
  difference that is smaller than that of $u$ and $v$, contradicting
  the choice of $e$.\qed
\end{proof}

An even shorter proof of the above can be obtained by reusing the
following observation from Cordovil and Forge~\cite{CF12}: for every
acyclic orientation of $G$, the set of flippable edges is a spanning
set of edges. Therefore, every such set must intersect every edge cut.

While guarding sets are hard to recognize, even for complete graphs,
we show that a minimum-size guarding set can be found in polynomial
time whenever the input graph is chordal. In that case, the upper
bound given by the minimum edge cut is tight, and the result
generalizes Lemma~\ref{lem:complete}.

\begin{theorem}
The minimum size of a guarding set in a chordal graph is the size of a minimum edge cut.
\end{theorem}
\begin{proof}
  We need to show that whenever a set $S$ of edges of a chordal graph
  $G$ has size strictly less than the edge connectivity of $G$, there
  exists an acyclic orientation of $G$ in which all edges of $S$ are
  transitive. Let us denote by $k$ the edge connectivity of $G$, and
  $n$ its number of vertices.

  We proceed by induction. The base case consists of a complete graph
  on $k+1$ vertices. From Lemma~\ref{lem:complete}, the minimum size
  of a guarding set is $k$.  Now suppose the statement holds for every
  chordal graph with $n-1$ vertices, and that there exists a $k$ edge
  connected chordal graph $G$ on $n$ vertices with a guarding set $S$
  of size $k-1$.  Since $G$ is chordal, it has at least two
  nonadjacent simplicial vertices $u$ and $v$. The degree of both $u$
  and $v$ is at least $k$. Hence one of them, say $v$, is incident to
  at most $\lfloor (k-1)/2\rfloor \leq \lfloor (d(v)-1)/2\rfloor$
  edges of $S$. Now remove $v$ and consider, using the induction
  hypothesis, a suitable acyclic orientation of the remaining
  graph. This orientation induces a total order, i.e., a path $p$ with
  $d(v)$ vertices, on the neighbors $N(v)$ of $v$.  Then vertex $v$
  must have one or two incident edges that do not belong to $S$ and
  connect $v$ to the first, or the last, or two consecutive vertices
  of path $p$. Hence we can integrate $v$ in the path such that all
  the edges of $S$ that are incident to $v$ are transitive. This
  yields a suitable acyclic orientation for $G$ and completes the
  induction step.\qed
\end{proof}

When the graph is not chordal, it may happen that the minimum size of
a guarding set is arbitrarily small compared to the edge
connectivity. We can in fact construct a large family of such
examples.

\begin{theorem}
  For every natural number $t\geq 2$ and odd natural number $g$ such
  that $3 \leq g\leq t$, there exists a graph $G$ with edge
  connectivity $t$ and a guarding set of size $g$.
\end{theorem}

\begin{wrapfigure}[12]{r}{0.44\textwidth}%
\vskip-9mm
\quad\includegraphics[width=0.36\textwidth]{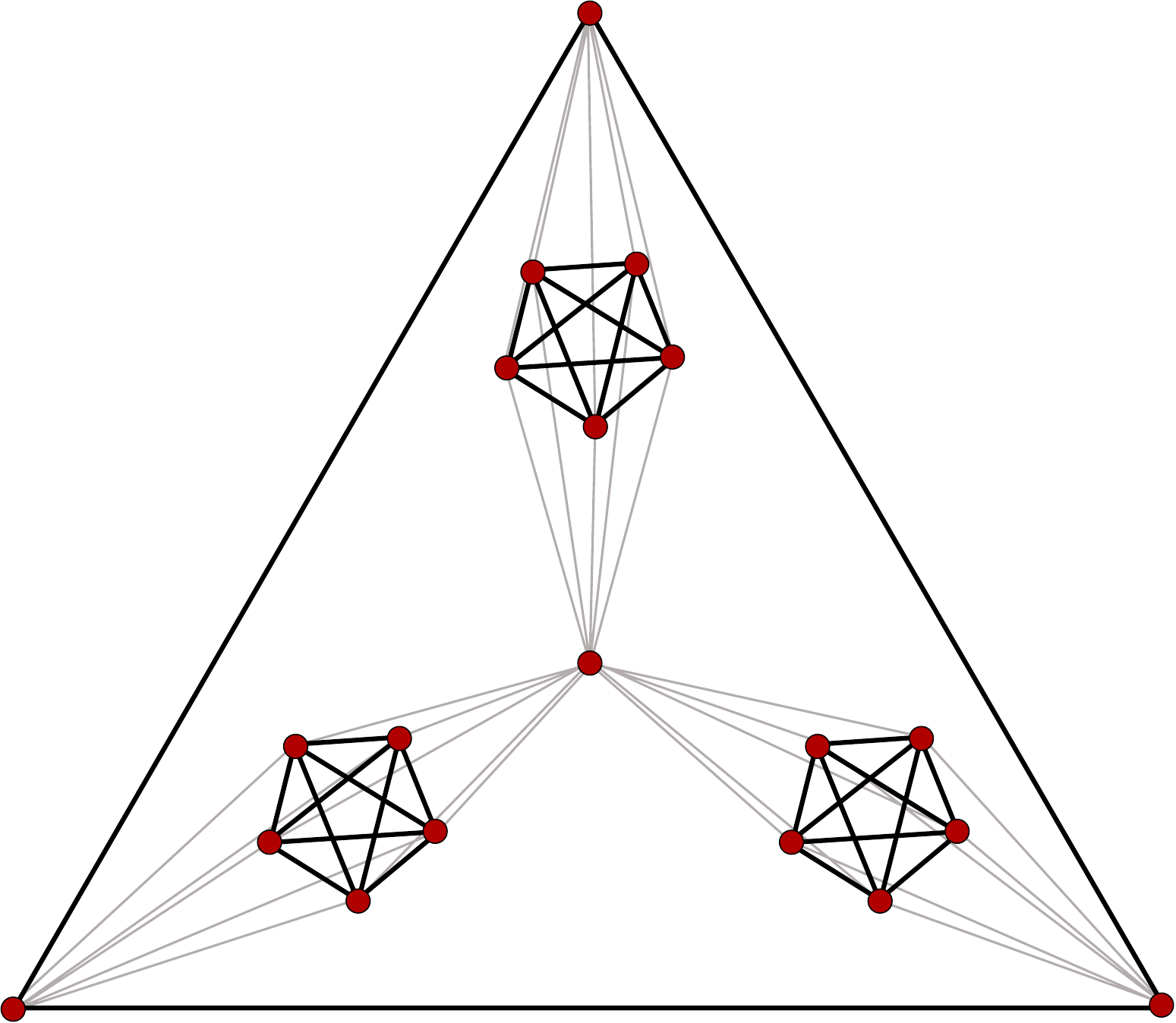}
\caption{\label{fig:exconn}A graph $G$ with edge connectivity 6 and a guarding
  set of size 3.}
\end{wrapfigure}
\noindent{\it Proof.}  
The graph $G$ is constructed by considering a wheel with $g+1$
  vertices and center $c$, and replacing every edge incident to the
  center $c$ by a copy of the complete graph $K_{t-1}$ such that each vertex
  of the complete graph is connected to both endpoints of the original
  wheel edge. Figure~\ref{fig:exconn} shows an example. Let us call $C_g$
  the cycle induced by the non-center vertices.

  Let us first look at the edge connectivity of $G$. Vertices
  belonging to one of the $K_{t-1}$ have $t$ incident edges. On the
  other hand we easily construct $t$ edge-disjoint paths between every
  pair of vertices. Hence, the edge connectivity of $G$ equals $t$.

  We now show that the set of vertices of $C_g$ is a guarding set.

  Suppose the opposite: there exists an acyclic orientation for which
  no edge of the cycle is flippable, hence all of them are transitive.
  Since $C_g$ is odd, there must exist two consecutive edges of $C_g$
  with the same orientation, say $uv$ and $vw$,
 with $uv$ oriented
  from $u$ to $v$, and $vw$ from $v$ to $w$. 

The only way to make $vw$
  transitive is to construct an oriented path from $v$ to $w$ going
  through the center $c$. Hence there must exist an oriented path of
  the form $vPc$, where $P$ is a path in the complete graph $K_{t-1}$
  attached to $v$. To make $uv$ transitive we need a directed path
  $cP'v$. The oriented cycle $vPcP'v$ is in contradicion to the acyclicity.  
  Therefore, some edge of $C_g$ is always flippable.\qed

\section{Median Graphs}
\begin{wrapfigure}[15]{r}{0.44\textwidth}%
\vskip-9mm
{\centering
\includegraphics[width=.44\textwidth]{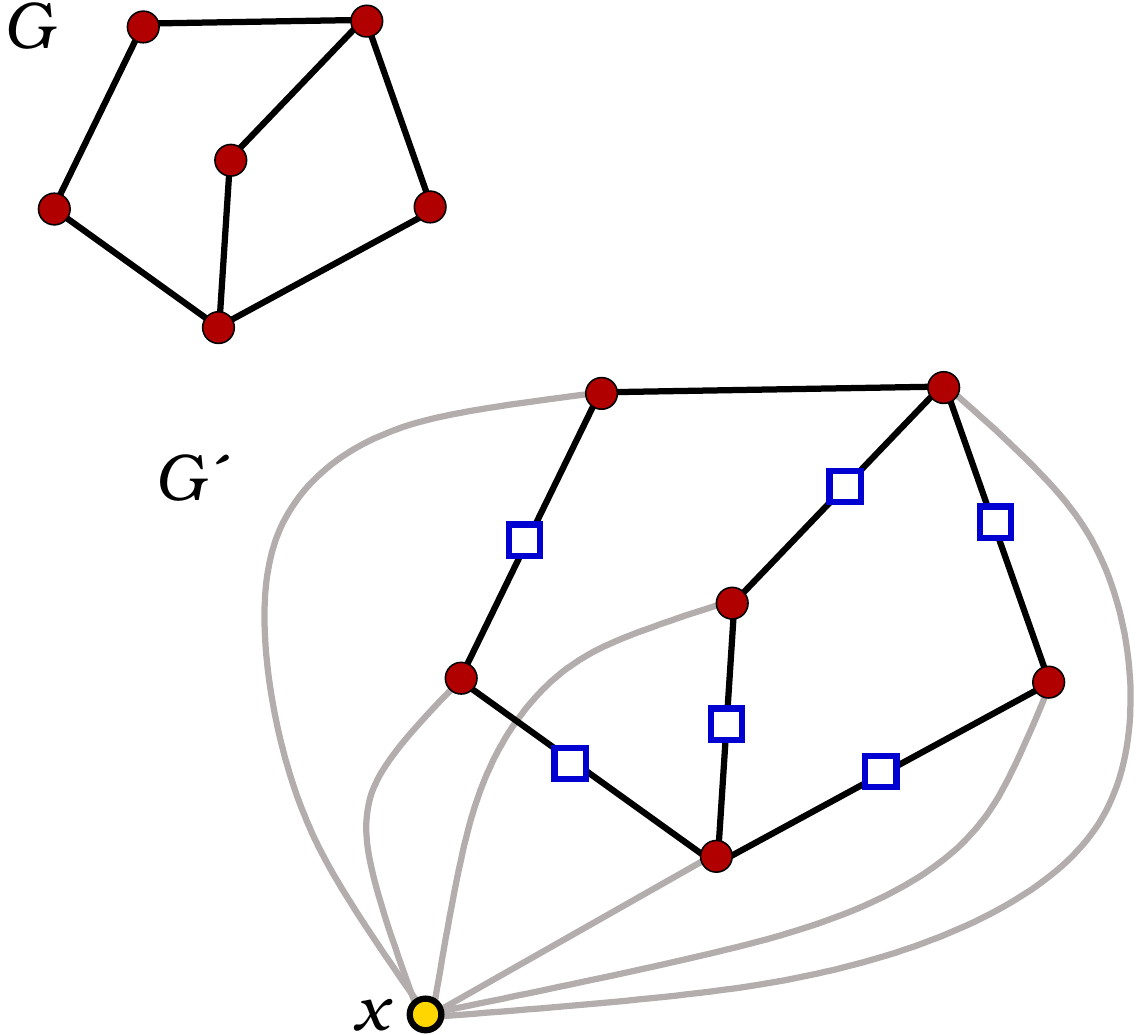}
\caption{A graph $G$ and the $G'$ obtained by the construction.} 
\label{fig:medgraph}}
\end{wrapfigure}

A median graph is an undirected graph in which every three vertices $x$, $y$,
and $z$ have a unique {\it median}, i.e., a vertex $\mu(x,y,z)$ that belongs
to shortest paths between each pair of $x$, $y$, and $z$. Imrich, Klav\v{z}ar,
and Mulder~\cite{IKM99} proposed the following construction of median graphs:
Start with any triangle-free graph $G=(V,E)$. First add an apex vertex $x$
adjacent to all vertices of $V$, then subdivide every edge of $E$ once, and let
$G'$ be the resulting graph. Figure~\ref{fig:medgraph} illustrates the construction.

We only need that $G'$ is a partial cube. This can be shown with the explicit
construction of an isometric embedding into $Q_n$. Let $V=\{v_1,\ldots,v_n\}$
and map these vertices bijectively to the standard basis, i.e., $v_i \to
e_i$. Apex $x$ is mapped to $0$ and the subdivision vertex $w_{ij}$ of an edge
$v_iv_j$ is mapped to $e_i+e_j$. The embedding shows that the zones are in
one-to-one correspondence with the vertices of $G$, where the zone of $v_i$ consists
of the edge $xv_i$ together with all edges $w_{ij}v_{j}$. Hence, a zone cover
of $G'$ corresponds to a subset of $V(G)$.  Let us call $w_e$ the vertex of
$G'$ subdividing the edge $e\in E(G)$. The construction of the embedding into
$Q_n$ shows that the transformation $G \to G'$ yields a partial cube even if
$G$ contains triangles, only membership in the smaller class of median graphs
is lost in this case.

\begin{lemma}
\label{lem:medfromtri}
If $G$ has no isolated vertices, then
the minimum size of a zone cover of $G'$ equals the minimum
size of a minimum vertex cover of $G$.
\end{lemma}
\begin{proof}
  We first show that a vertex cover $S$ of $G$ corresponds to a zone
  cover of $G'$ of the same size.  For a vertex $w_{ij}$ at least one
  of $v_i$ and $v_j$ is in $S$, hence $w_{ij}$ is covered. If $v_i$
  belongs to $S$ then it is covered by its own zone. Otherwise there
  is an edge $v_iv_j$ and necessarily $v_j\in S$, therefore in this case $v_i$ is
  covered by the zone of $v_j$. The apex $x$ is covered by every zone.

  The other direction is straightforward. A zone cover of $G'$ must in
  particular cover all subdivision vertices $w_e\in V(G')$. Hence the zone
  corresponding to at least one of the endpoint of $e$ must be
  selected, yielding a vertex cover in $G$.\qed
\end{proof}

Given a connected triangle-free graph $G$, one can construct the median graph
$G'$ in polynomial time. Since deciding whether $G$ has a vertex cover of size at 
most $k$ is NP-complete, even on triangle-free graphs, we obtain:

\begin{corollary}
\label{cor:medfromtri}
Given a median graph $G'$ and an integer $k$, deciding whether there
exists a zone cover of size at most $k$ of $G'$ is NP-complete.
\end{corollary}

The minimum vertex cover problem is hard to
approximate, even on triangle-free graphs. This directly yields the
following corollary.

\begin{corollary}
\label{cor:medfromtriapx}
Given a median graph $G'$, finding a minimum size zone cover of $G'$ is
APX-hard.
\end{corollary}


\section{Distributive Lattices}

In this special case of the partial cube covering problem, we
are given a poset $P$, and we wish to find a subset $S$ of its
elements such that the following holds: for every downset $D$ of $P$,
there exists $x\in S$ such that either $D\cup\{x\}$ or $D\setminus
\{x\}$ is a downset, distinct from $D$.  Given a poset $P$, we
refer to a suitable set $S$ as a {\em guarding set} for $P$, and let
$g(P)$ be the size of a smallest guarding set.
Figure~\ref{fig:poset} gives an example of poset and the corresponding partial cube. 

\subsection{Relation to poset fibres}

We first establish a connection between guarding sets and fibres. 
A {\em fibre} of a poset is a subset of its elements that meets every nontrivial maximal antichain. 
Let $f(P)$ be the size of a smallest fibre of $P$.

\begin{lemma}
\label{lem:fibre}
Every fibre is a guarding set. In particular, $g(P)\leq f(P)$.
\end{lemma}
\begin{proof}
  Consider a poset $P=(V,\leq )$ and one of its downset $D\subseteq
  V$. Let $F := V\setminus D$, $A:=\max (D)$, and $B:=\min (F)$.  By
  definition, a guarding set is a hitting set for the collection of
  subsets $A\cup B$ constructed in this way.

  Let us now consider the subset $A\cup B'\subseteq A\cup B$, where
  $B' := \{ b\in B: b \text{\ is\ incomparable\ to\ } a, \forall a\in
  A \}$. This set is easily shown to be a maximal antichain. Hence, it
  must be hit by any fibre.\qed
\end{proof}

Duffus, Kierstead, and Trotter~\cite{DKT91} have shown that
every poset on $n$ elements has a fibre of size at most $2n/3$. This
directly yields the following.
\begin{corollary}
\label{cor:fibre}
For every $n$-element poset $P$, $g(P)\leq 2n/3$.
\end{corollary}

We now consider a special case for which the notions of guarding set and fibre
coincide.  A poset is bipartite if $P=\min(P) \cup \max(P)$, i.e., if the
height of $P$ is at most 2.

\begin{lemma}
\label{lem:bip}
For a bipartite poset $P$, a set $S$ is a guarding set for $P$ if and only if it is a fibre of $P$.
\end{lemma}
\begin{proof}
  We know from Lemma~\ref{lem:fibre} that every fibre is a guarding
  set. The other direction is as follows.

  Consider a guarding set $S$ and let $A$ be any maximal antichain of $P$. Let
  $T := A\cap \max (P)$, and let $D$ be the downset generated by $T$. As a
  downset $D$ is guarded, hence either an element of $T$ is hit by $S$, or an
  element of $\min (P\setminus D)$ is hit, but since $A$ is maximal $A=T \cup
  \min (P\setminus D)$, hence $A$ is hit. Therefore, $S$ is a fibre.\qed
\end{proof}

\subsection{Complexity}

Lemma~\ref{lem:bip} yields two interesting corollaries on the
complexity of recognizing and finding guarding sets.

\begin{corollary}
\label{cor:recposet}
Given a poset $P$ and a subset $S$ of its elements, the problem of
deciding whether $S$ is a guarding set is coNP-complete.  This holds
even if $P$ is bipartite.
\end{corollary}
\begin{proof}
  Recognition of fibres in bipartite posets has been proved
  coNP-complete by Duffus et al.~\cite{DKT91}.  From
  Lemma~\ref{lem:bip}, this is the same problem as recognizing
  guarding sets.\qed
\end{proof}

\begin{corollary}
\label{cor:findposet}
Given a poset $P$ and an integer $k$, the problem of deciding whether
there exists a guarding set of size at most $k$ is
$\Sigma_2^P$-complete. This holds even if $P$ is bipartite.
\end{corollary}
\begin{proof}
  Again, this is a consequence of Lemma~\ref{lem:bip} and a recent
  result of Cardinal and Joret~\cite{CJ12} showing that the
  corresponding problem for fibres in bipartite posets is
  $\Sigma_2^P$-complete.\qed
\end{proof}

Duffus et al.~\cite{DSSW91} mention that it is possible to construct
posets on $15n+2$ elements, every fibre of which must contain at least
$8n+1$ elements, which gives a lower bound with a factor $8/15$. The
guarding sets for these examples do not seem to require as many
elements.

\section*{Open Problems}
We left a number of problems open. For instance, we do not know the complexity
status of the problem of deciding whether there exists a guarding set of edges
of size at most $k$ in a given graph. A natural candidate class would be
$\Sigma_2^P$. For the same problem, we do not have any nontrivial lower bound
on the minimum size of a guarding set. It would also be interesting to give
tighter lower and upper bounds on the minimum size of a guarding set in a
poset and in a line arrangement. Finally, questions involving
partial cubes derived from antimatroids, such as elimination orderings in chordal
graphs, are currently under investigation.

\paragraph{Acknowledgments.}
This work was initiated at the 2nd ComPoSe Workshop held at the TU Graz,
Austria, in April 2012. We thank the organizers and the participants for the
great working atmosphere. We also acknowledge insightful discussions on
related problems with several colleagues, in particular Michael Hoffmann (ETH
Z\"urich) and Ferran Hurtado (UPC Barcelona).

\bibliographystyle{siam}
\global\advance\itemsep12pt
\bibliography{cubecover}

\end{document}